\documentclass[10pt,journal]{IEEEtran}

\usepackage[cm]{fullpage}
\usepackage[dvips]{graphicx}
\usepackage[usenames,dvipsnames]{xcolor} 
\usepackage{tikz}
\usetikzlibrary{shapes,arrows,positioning,calc}
\usepackage{amsmath,amssymb, amsthm}
\usepackage{xfrac}
\usepackage{epsfig}
\usepackage{psfrag}
\usepackage{enumitem}
\usepackage{comment}
\usepackage{ushort}
\usepackage{stfloats}
\usepackage{subfigmat}

\providecommand{\abs}[1]{|#1|}
\providecommand{\norm}[1]{\lVert#1\rVert}

\newcommand*\oline[1]{%
  \vbox{%
    \hrule height 0.45pt %
    \kern0.19ex%
    \hbox{%
      \kern-0.225em%
      \ifmmode#1\else\ensuremath{#1}\fi%
      \kern-0.01em%
    }
  }
}

\newtheorem{thm}{Theorem}

\newtheorem{lem}{Lemma}
\newtheorem{defn}{Definition}
\newtheorem{prop}{Proposition}

\providecommand{\Bupp}{\oline{B}}
\providecommand{\B}{\ushort{B}}
\providecommand{\BU}{\ushort{B}_U}
\providecommand{\BL}{\ushort{B}_L}

\providecommand{\marker}{0.4}
\providecommand{\pvar}{0.7}
\providecommand{\plab}{0.7}
\providecommand{\intc}{0.5}
\providecommand{\figwid}{8.4cm}

\DeclareMathOperator{\Proj}{Proj}

\title{Adaptive Control of Scalar Plants in the Presence of Unmodeled Dynamics}

\author{Heather~S.~Hussain, Megumi~M.~Matsutani, Anuradha~M.~Annaswamy and
       Eugene~Lavretsky
\thanks{H.~H. Hussain and A. M. Annaswamy are with the Department
of Mechanical Engineering, Massaschusetts Institute of Technology, Cambridge,
MA, 02139 e-mail: ({hhussain@mit.edu}).}
\thanks{M.~M. Matsutani is with the Department
of Aeronautics and Astronautics, Massaschusetts Institute of Technology, Cambridge,
MA, 02139}
\thanks{E. Lavretsky is with The Boeing Company, Huntington Beach, CA
92648}
\thanks{This work is supported by the Boeing Strategic University Initiative.}}

\begin{document}

\maketitle

\begin{abstract}                
Robust adaptive control of scalar plants in the presence of unmodeled dynamics is established in this paper. It is shown that implementation of a projection algorithm with standard adaptive control of a scalar plant ensures global boundedness of the overall adaptive system for a class of unmodeled dynamics.
\end{abstract}

\section{Introduction} \label{sec:1}
Following the Rohrs counterexample in \cite{Rohrs}, several robust adaptive control solutions were suggested in the '80s and '90s (see, for example, \cite{1} and \cite{5}), including specific responses to the counterexample (see for example \cite{3}, \cite{7}, \cite{11}, \cite{5}, and \cite{10}). Most of these were qualitative, or local, and often involved properties of persistent excitation of the reference input. In this paper, we show that for a class of unmodeled dynamics including the one in \cite{Rohrs}, adaptive control of a scalar plant with global boundedness can be established for any reference input.
\tikzstyle{block} = [draw, fill=Aquamarine!50, rectangle, 
    minimum height=1cm, minimum width=2cm, rounded corners=0.5mm, inner sep= 2mm]
\tikzstyle{sum} = [draw, fill=Dandelion!40, circle, minimum width=1cm, inner sep= 1mm]
\tikzstyle{feedback} = [draw, fill=Green!40, circle, minimum width=1cm]
\tikzstyle{input} = [coordinate]
\tikzstyle{output} = [coordinate]
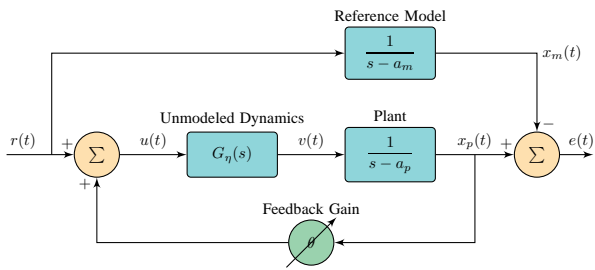
\begin{figure}[h!]
\begin{center}
\begin{tikzpicture}[auto, scale=0.6, every node/.style={transform shape}, node distance=1.0cm, >=latex']
\node [input] (input1) {};
\node [input, right of= input1] (input2) {};    
\node [sum, right of=input2] (sum1) {$\sum$} ;
\node [block, right of=sum1, label=above: {Unmodeled Dynamics}, node distance=3.0cm] (block1) {$G_{\eta}(s)$};
\node [block, right of=block1, label=above:{Plant},
            node distance=3.5cm] (block2) {$\displaystyle \frac{1}{s-a_p}$};
\node [block, above of=block2, label=above:{Reference Model}, node distance=2.25cm] (block3) {$\displaystyle \frac{1}{s-a_m}$};
\node [sum, right of=block2,node distance=3.25cm] (sum2) {$\sum$};
\node [output, right of=sum2,node distance=1.25cm] (output1) {};

    \draw [->] (block1) -- node[name=u] {$v(t)$} (block2);
    \node [output, right of=block2] (output2) {};
    \node [feedback, below of=u, label=above:{Feedback Gain}, node distance=2.25cm] (block4) {${\theta}$};
	\draw [->] (input1) -- node[near start]{$r(t)$} node[pos=0.9] {$+$} (sum1);
    \draw [->] (sum1) -- node {$u(t)$} (block1);
    \draw [->] (block2) -- node [name=y] {$x_p (t)$} node[pos=0.9] {$+$}(sum2);
    \draw [->] (y) |- (block4);
    \draw [->] (block4) -| node[pos=0.95] {$+$} node [pos=0.075] {} (sum1);
    \draw [->] (input2) |- (block3);
	\draw [->] (block3) -| node {$x_m (t)$} node[pos=0.95] {$-$}(sum2);
	\draw [->] (sum2) --  node[pos=0.65]{$e(t)$}(output1);
		\draw[->] (6.2,-2.5)--($(block4.225)!1.3cm!(block4.45)$);
\end{tikzpicture}
\caption{Adaptive control in the presence of unmodeled dynamics} 
\label{fig:ControlProb}
\end{center}
\end{figure}

\section{The Problem Statement: Scalar Plant} \label{sec:2}
The problem we address in this paper is the adaptive control of a first-order plant 
\begin{equation} \label{eq:plant}
\dot{x}_p(t)= a_p x_p(t)+v(t)
\end{equation}
where $a_p$ is an unknown parameter. It is assumed that $|a_p| \leq \bar{a}$, where $\bar{a}$ is a known positive constant. The unmodeled dynamics are unknown and defined as 
\begin{equation}\label{eq:xetadot}
\begin{aligned} 
\dot{x}_\eta (t)&=A_\eta x_\eta (t) + b_\eta u(t) \\
v(t)&={c_\eta}^T x_\eta(t)
\end{aligned}
\end{equation}
where $A_\eta \in \mathbb{R}^{n \mathsf{x} n}$ is Hurwitz with 
\begin{equation} \label{eq:Geta}
G_\eta (s) \triangleq c_\eta^T(sI_{n \mathsf{x} n}-A_\eta)^{-1} b_\eta\text{.}
\end{equation}
$x_\eta(t)$ is the state vector, and $u(t)$ is the control input. The goal is to design the control input such that $x_p(t)$ follows $x_m(t)$ which is specified by the reference model
\begin{equation}  \label{eq:refmod}
\dot{x}_m(t)= a_m x_m(t)+r(t)
\end{equation}
where $a_m < 0$, and $r(t)$ is the reference input. The adaptive controller we propose is a standard adaptive control input given by (see figure \ref{fig:ControlProb})
\begin{equation} \label{eq:controlin}
u(t)=\theta(t)x_p(t)+r(t)
\end{equation}
where the parameter $\theta(t)$ is updated using a projection algorithm given by
\begin{equation} \label{eq:adaptlaw}
\dot{\theta}(t)=\gamma \Proj(\theta(t),-x_p(t)e(t)), \; \gamma >0
\end{equation}
where
\begin{equation}
e(t)=x_p(t)-x_m(t) \label{eq:e}
\end{equation}
\begin{align} \label{eq:Proj}
\Proj(& \theta,y)=
 \begin{cases} 
 \displaystyle \frac{\theta_{max}^2-{\theta^2}}{{\theta_{max}^2}-{\theta_{max}^{\prime 2}}}y & [\theta \in \Omega_A, \; y\theta>0]\\\\
y &\text{otherwise}
\end{cases}
\end{align}
\begin{align} \label{eq:omegasets}
\Omega_0 &= \lbrace\theta \in \mathbb{R}^1 \; \lvert \; -\theta_{max}^{\prime} \leq \theta \leq \theta_{max}^{\prime} \rbrace \nonumber \\
\Omega_1 &= \lbrace \theta \in \mathbb{R}^1 \; \lvert \; -\theta_{max} \leq \theta \leq \theta_{max}\rbrace \\
\Omega_A&=\Omega_1 \backslash \Omega_0 \nonumber
\end{align}
with positive constants $\theta_{max}^{\prime}$ and $\theta_{max}$ given by
\begin{equation}  \label{eq:thetaminor}
\theta_{max}^{\prime} >\bar{a}+ |a_m|
\end{equation}
\begin{equation} \label{eq:thetamajor}
\theta_{max}=\theta_{max}^{\prime} +\varepsilon_0, \quad \varepsilon_0 >0.
\end{equation}
\begin{lem}\label{lem:1}
 Consider the Adaptive Law in (\ref{eq:adaptlaw}) with Projection Algorithm in (\ref{eq:Proj}) to (\ref{eq:thetamajor}). Then,
\begin{equation} \label{eq:lem1}
\norm{\theta(t_a)} \leq \theta_{max}  \Longrightarrow \norm{\theta(t)} \leq \theta_{max}, \; \forall t \geq t_a \text{.}
\end{equation}
Hence, the projection algorithm guarantees the boundedness of
the parameter $\theta(t)$ independent of the system dynamics. We refer the reader to \cite{Proj} for the proof of Lemma \ref{lem:1}.
\end{lem}
\section{Choice of Projection Parameters} \label{sec:3}
The projection algorithm in (\ref{eq:Proj}) is specified by two parameters $\theta_{max}^{\prime} $ and $\theta_{max}$. Equation (\ref{eq:thetaminor}) provides the condition for $\theta_{max}^{\prime} $. To determine $\varepsilon_0$ in (\ref{eq:thetamajor}), the following discussions are needed:

We consider the linear time-invariant system specified by (\ref{eq:plant}), (\ref{eq:xetadot}), and (\ref{eq:controlin}), with the parameter $\theta(t)$ fixed as
\begin{equation} \label{eq:thetafixed}
\theta(t)=-\theta_{max}, \quad \forall t \geq t_a \text{.}
\end{equation}
The closed loop transfer function from $r(t)$ to $x_p(t)$ is given by
\begin{equation} \label{eq:Gcl}
G_c (s)= \frac{p_\eta (s)}{q_c (s)}
\end{equation}
where $G_c(s)$ is defined using $G_\eta(s)$ in (\ref{eq:Geta}) as
\begin{align}
G_\eta (s) & \triangleq \frac{p_\eta (s)}{q_\eta (s)} \label{eq:Getacl} \\
q_c(s) & = q_\eta (s) (s-a_p)+ \theta_{max} p_\eta (s). \label{eq:qc}
\end{align}
From (\ref{eq:thetaminor}) and (\ref{eq:thetamajor}), it follows that
\begin{equation}
a_p-\theta_{max}<0, \quad \forall \abs{a_p} \leq \bar{a}.
\end{equation}
Therefore it follows that there exists a class of unmodeled dynamics $(c_\eta, A_\eta, b_\eta)$ such that $q_c (s)$ has roots in $\mathbb{C}^-$, the left-half of the complex plane. It is this class that is of interest in this paper.
\begin{defn} \label{def:1}
The triple $(c_\eta, A_\eta, b_\eta)$ is said to belong to \newline $S_\eta(\bar{a}, \theta_{max})$ if $p_\eta (s)$ and $q_\eta(s)$ in (\ref{eq:Getacl}) are such that the roots of $q_c(s)$ in (\ref{eq:qc}) lie in $\mathbb{C}^-$ for all $|a_p|\leq \bar a$.
\end{defn}
Let's demonstrate $S_\eta(\bar{a}, \theta_{max})$ with an example. Consider the class of unmodeled dynamics of the form
\begin{equation} \label{eq:Geta2ndorder}
G_\eta (s) = \frac{\omega_n^2}{s^2+2\zeta\omega_ns +\omega_n^2}
\end{equation}
where $\zeta>0$ and $\omega_n>0$. From (\ref{eq:Gcl}), (\ref{eq:qc}), and (\ref{eq:Geta2ndorder}), the closed-loop dynamics from $r$ to $x_p$ is given by
\begin{equation} \label{eq:Gcl_2ndorder}
G_c (s)= \frac{\omega_n^2}{q_c (s)}
\end{equation}
where
\begin{equation} \label{eq:qc_Gcl}
\begin{aligned}
q_c(s) & = s^3 + a_1 s^2 + a_2 s +a_3 \\
a_1&= (2 \zeta \omega_n - a_p)\\
a_2&= (\omega_n^2-2a_p \zeta \omega_n) \\
a_3&= -a_p \omega_n^2+\theta_{max}\omega_n^2
\end{aligned}
\end{equation}
For the roots of $q_c(s)$ in (\ref{eq:qc_Gcl}) to lie in $\mathbb{C}^-$, the following conditions are neccessary and sufficient for all $|a_p|\leq \bar a$:
\begin{enumerate}[itemindent=0.1in, label=(A-{\roman*})]
\item $a_p < \min(2 \zeta \omega_n,\displaystyle \frac{\omega_n}{2\zeta})$ \label{cond1}
\item $ \theta_{max} > a_p$ \label{cond2}
\item $\theta_{max} < \displaystyle \bigg(-4a_p\zeta^2 + \frac{2\zeta a_p^2}{\omega_n} + 2\zeta\omega_n \bigg)$ \label{cond3}
\end{enumerate}
If 
\begin{equation} \label{eq:ineq1}
a_p < \theta_{max} <a_p+ \bar\theta^\star
\end{equation}
where
\begin{equation} \label{eq:cond3}
\bar\theta^\star=(2\zeta\omega_n-a_p)(1-\frac{2\zeta a_p}{\omega_n})
\end{equation}
then conditions \ref{cond2} and \ref{cond3} hold.

Hence, any class of unmodeled dynamics $(c_\eta, A_\eta, b_\eta)$ in (\ref{eq:Geta2ndorder}) satisfying condition \ref{cond1} belongs to $S_\eta(\bar{a}, \theta_{max})$. It can be easily shown that the unmodeled dynamics and the plant discussed in the infamous Rohrs counterexample \cite{Rohrs} satisfies conditions \ref{cond1} to \ref{cond3} above for some $\theta_{max}$.

We now discuss the choice of $\varepsilon_0$. Consider the class of unmodeled dynamics $S_\eta(\bar{a}, \theta_{max})$ in Definition \ref{def:1}. Since the closed loop system specified by (\ref{eq:plant}), (\ref{eq:xetadot}), (\ref{eq:controlin}), and (\ref{eq:thetafixed}) is stable, it follows that there exists a Lyapunov function
\begin{equation} \label{eq:lyap}
V= \bar{x}^T P \bar{x}
\end{equation}
with a time derivative
\begin{equation} \label{eq:lyapderiv}
\dot{V}= -\bar{x}^T Q \bar{x}
\end{equation}
where $\bar{x}=[x_p \; x_\eta^T]^T$. $P$ is the solution to the Lyapunov equation
\begin{equation} \label{eq:lyap2PQ}
\bar{A}^T P + P \bar{A} = -Q < 0
\end{equation}
where
\begin{equation}
\bar{A}=
\begin{bmatrix} 
a_p & c_\eta^T \\ 
-b_\eta \theta_{max} & A_\eta
\end{bmatrix}\
\end{equation}
since $\bar{A}$ is Hurwitz. The latter is true since $\theta_{max}$ satisfies (\ref{eq:ineq1}).

We define two sets $\Omega_u \subset \Omega_A$ and $\Omega_l \subset \Omega_A$ as
\begin{align}
\Omega_u & = \lbrace \theta \in \mathbb{R}^1 \; \lvert \; -\theta_{max}+\xi_0 \leq \theta < -\theta_{max}^{\prime} \rbrace \label{eq:omegau}\\
\Omega_l & = \lbrace\theta \in \mathbb{R}^1 \; \lvert \; -\theta_{max} \leq \theta \leq -\theta_{max}+\xi_0 \rbrace \label{eq:omegal}
\end{align}
where \begin{equation} \label{eq:varth}
\xi_0=c \varepsilon_0, \quad c \in (0,1).
\end{equation}
We now consider the linear time-varying system specified by (\ref{eq:plant}), (\ref{eq:xetadot}), and (\ref{eq:controlin}), with $\theta(t) \in \Omega_u \cup \Omega_l$. It follows from (\ref{eq:thetamajor}) and (\ref{eq:lem1}) that
\begin{align}
\theta(t) &=-\theta_{max}+\varepsilon(t), \quad \forall \theta(t) \in \Omega_u \cup \Omega_l \label{eq:theps} \\
\theta(t) &=-\theta_{max}+\xi(t), \quad \forall \theta(t) \in \Omega_l	\label{eq:thevar}
\end{align}
where
\begin{equation} \label{eq:varth_t}
\varepsilon(t) \in [0,\varepsilon_0), \quad \xi(t) \in [0,\xi_0].
\end{equation}
Therefore, the closed-loop system is given by
\begin{equation} \label{eq:LTVsys}
\dot{\bar{x}}=\bar{A}\bar{x}+A_\xi(t)\bar{x}+\bar{b}r, \quad \forall \theta(t) \in \Omega_l
\end{equation}
where
\begin{equation} \label{eq:LTV_A}
A_\xi (t)=
\begin{bmatrix} 
0 & 0 \\ 
b_\eta \xi(t) & 0
\end{bmatrix}, \quad \bar{b}=
\begin{bmatrix} 
0 \\ 
b_\eta
\end{bmatrix}.
\end{equation}
If we choose $V=-\bar{x}^T Q \bar{x}$ with $P$ as in (\ref{eq:lyap2PQ}), we obtain
\begin{equation} \label{eq:lyapderivLTV}
\dot{V} \leq - \lambda_{Q_{min}}\norm{\bar{x}}^2+2\lambda_{P_{max}} k \xi_0 \norm{\bar{x}}^2 +2\lambda_{P_{max}}\norm{\bar{b}}r_{max}\norm{\bar{x}} \\
\end{equation}
where
\begin{align} \label{eq:lamdamax}
\begin{aligned}
\lambda_{Q_{min}} & \triangleq \min\limits_{i} |\Re (\lambda_i(Q))| \\
\lambda_{P_{max}} & \triangleq \max\limits_{i} |\Re (\lambda_i(P))|
\end{aligned}
\end{align}
\begin{equation}
\norm{b_\eta} \leq k, \quad r_{max}=\max\limits_{t \geq t_a} \abs{r(t)}.
\end{equation}
That is, 
\begin{equation} \label{eq:lyap_negdef}
\dot V < 0 \quad {\rm if} \; \norm{\bar{x}} > x_0
\end{equation}
where
\begin{align}
x_0 & =\frac{2\lambda_{P_{\max}}\norm{\bar b}r_{\max}}{\bar\lambda} \label{eq:x0} \\
\bar\lambda & =\lambda_{Q_{min}}-2\lambda_{P_{max}} k\xi_0. \label{eq:lamdabar}
\end{align}
In summary, the closed-loop system has bounded solutions for all $\theta(t) \in \Omega_l$ with $\norm{x(t)} \leq x_0$ if $(c_\eta, A_\eta, b_\eta)$ is such that
\begin{enumerate}[itemindent=0.1in, label=(B-{\roman*})]
\item $q_c(s)$ has roots in $\mathbb{C}^-$ for all $|a_p|\leq \bar a$, and \label{cond:1}
\item $\xi_0<\varepsilon_0$ \label{cond:2}, where
\item $\displaystyle \xi_0 < \frac{\lambda_{Q_{min}}}{2k\lambda_{P_{max}}}$ \label{cond:3}
\end{enumerate}
We introduce the following definition:
\begin{defn} \label{def:2}
The triple $(c_\eta, A_\eta, b_\eta)$ is said to belong to \newline $S_\eta(\bar{a},\theta_{max},\xi_0)$ if conditions \ref{cond:1}, \ref{cond:2}, and \ref{cond:3} above are satisfied.
\end{defn}

\section{Main Result} \label{sec:4}
\begin{thm} \label{thm:1}
Let {${z(t)=[e(t) \; \theta(t)]^T}$}. The closed-loop adaptive system given by (\ref{eq:plant})-(\ref{eq:thetamajor}) has globally bounded solutions for all $\theta(t_a) \in \Omega_1$ if $(c_\eta, A_\eta, b_\eta) \in S_\eta(\bar{a},\theta_{max},\xi_0)$. 

\begin{defn} We define the region A and the boundary regions $\Bupp$ and $\B$ as follows
\begin{equation} \label{eq:boundaries}
\begin{aligned}
\Bupp &=\lbrace z \in \mathbb{R}^2 \; \lvert \; \theta_{max}^{\prime}  < \theta \leq \theta_{max} \rbrace \\
A &=\lbrace z \in \mathbb{R}^2 \; \lvert \; \theta \in \Omega_0  \rbrace \\
\B &=\lbrace z \in \mathbb{R}^2 \; \lvert \; -\theta_{max} \leq \theta < -\theta_{max}^{\prime}  \rbrace
\end{aligned}
\end{equation}
\end{defn}
\begin{defn}
We divide the boundary region $\B$ into two regions as follows:
\begin{equation} \label{eq:boundaries_lower}
\begin{aligned}
\BU &=\lbrace z \in \mathbb{R}^2 \; \lvert \; \theta \in \Omega_u \rbrace\\
\BL &=\lbrace z \in \mathbb{R}^2 \; \lvert \; \theta \in \Omega_l \rbrace
\end{aligned}
\end{equation}
with $\B=\BU \cup \BL$.
\begin{figure}[h!]
\psfrag{error}[t][][0.8]{Error, $e(t)$}
\psfrag{theta}[b][][0.8]{Parameter, $\theta(t)$}
\psfrag{z}[b][][0.8]{$z(t)=[e(t) \; \theta(t)]^T$}
\psfrag{A}[b][][0.7]{$A$}
\psfrag{Bbar}[][][0.6]{$\Bupp$}
\psfrag{BbarL}[][][0.6]{$\B$}
\psfrag{BU}[b,l][][0.6]{$\BU$}
\psfrag{BUarr}[l][][0.35]{$\rotatebox[origin=c]{90}{$\boldsymbol{\mapsto}$}$}
\psfrag{BL}[t,r][][0.6]{$\BL$}
\psfrag{BLarr}[r][][0.35]{$\rotatebox[origin=c]{-90}{$\boldsymbol{\mapsto}$}$}
\psfrag{zero}[b][][0.7]{$0$}
\psfrag{thm}[l][][0.7]{$\theta_{max}$}
\psfrag{thmp}[l][][0.7]{$\theta_{max}^{\prime}$}
\psfrag{minthm}[l][][0.7]{$-\theta_{max}$}
\psfrag{minthmp}[l][][0.7]{$-\theta_{max}^{\prime}$}
\psfrag{thstar}[l][][0.7]{$\theta^\star$}
\psfrag{vareps0}[r][][0.7]{$\varepsilon_0$}
\psfrag{vareps2}[r][][0.7]{$\xi_0$}
\psfrag{arr}[r][][0.75]{$\lbrace$}
\psfrag{arr2}[r][][0.4]{${\boldsymbol \lbrace}$}
\psfrag{ebar}[t][][0.7]{$\bar{e}$}
\psfrag{ebarmin2delta}[t][][0.7]{$\bar{e}-2\delta$}
\psfrag{ebarmindelta}[b][][0.7]{$\bar{e}-\delta$}
\psfrag{xmbar}[l,t][][0.7]{$\bar{x}_m$}
\psfrag{minxmbar}[r,t][][0.7]{$-\bar{x}_m$}
\psfrag{1}[t][][0.6]{\setlength{\tabcolsep}{0pt}\renewcommand{\arraystretch}{1.5}\begin{tabular}{c}{I}\end{tabular}}
\psfrag{2}[l][][0.6]{{II}}
\psfrag{3}[l][][0.6]{{III}}
\centering
\includegraphics[trim=3.1cm 1.25cm 0.5cm 1.3cm, clip=true, width=\figwid]{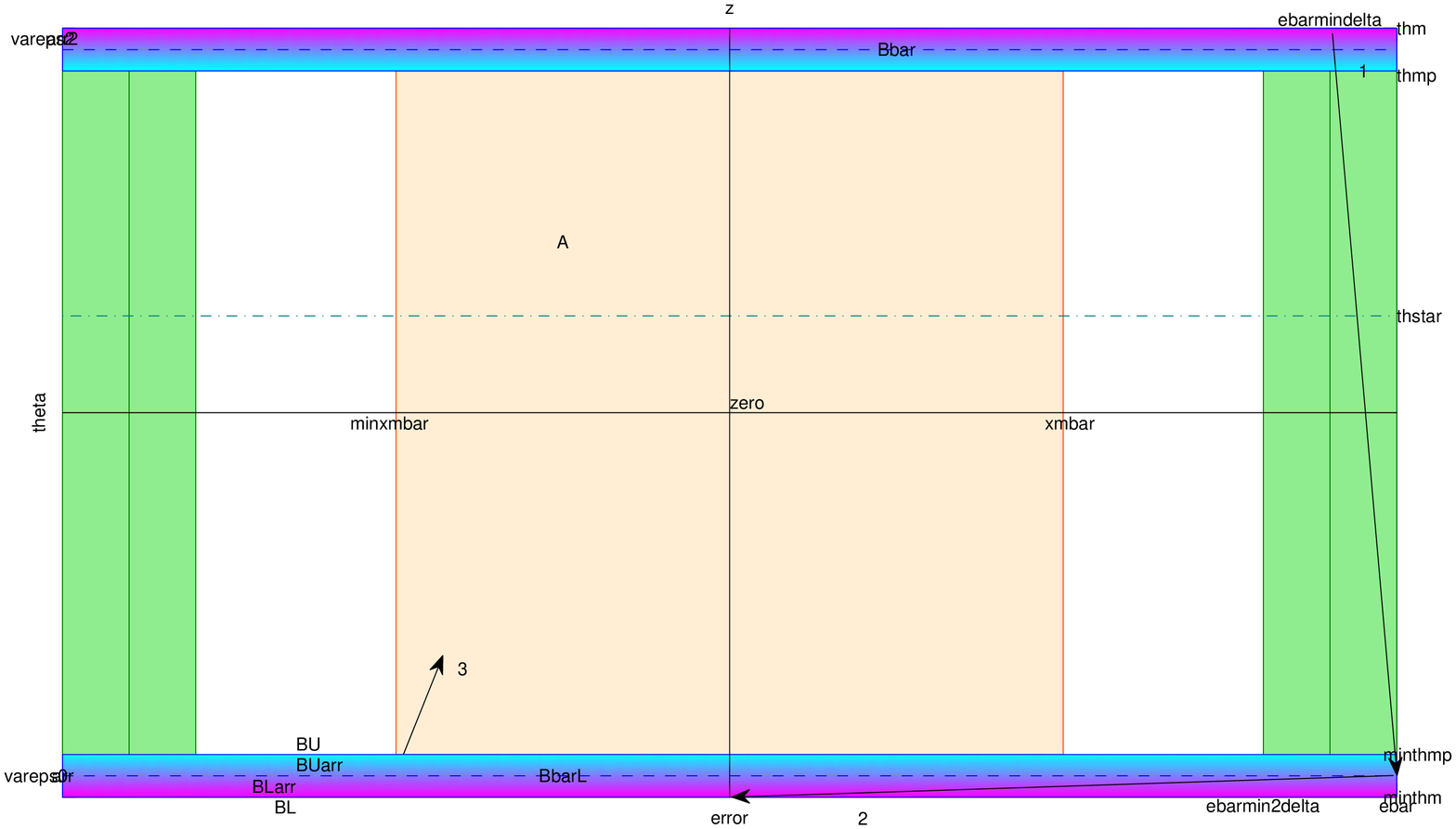} 
\caption{Definition of regions in (\ref{eq:boundaries}) and (\ref{eq:boundaries_lower}), and phases I-III.}
\label{fig:zspace}
\end{figure}
\end{defn}
\begin{proof}[Proof of Theorem \ref{thm:1}]
The closed-loop adaptive system has error dynamics in (\ref{eq:e}) equivalent to\footnote{For ease of exposition, we suppress the argument "t" in what follows.} 
\begin{equation} \label{eq:edot}
\dot{e}=a_m e+\widetilde{\theta}x_p+\eta
\end{equation}
where
\begin{equation} \label{eq:parerr}
\widetilde{\theta} =\theta - \theta^\star,\quad \theta^\star= a_m-a_p, \quad \eta = v - u.
\end{equation}
By combining the adaptive law in (\ref{eq:adaptlaw}) and (\ref{eq:Proj}), and boundary region definitions in (\ref{eq:boundaries}), we obtain
\begin{equation} \label{eq:adaptlawupdate}
\begin{aligned}
\dot{\theta}&=
& \begin{cases} 
 \displaystyle -\frac{\theta_{max}^2-{\theta^2}}{{\theta_{max}^2}- {{{\theta_{max}^{\prime 2}}}}} \gamma e x_p & \text{if} \; z \in (\B \cup \Bupp), \; -e x_p \theta >0\\\\
- \gamma e x_p & \text{otherwise}
\end{cases}
\end{aligned}
\end{equation}
Hence, the projection algorithm guarantees the boundedness of the parameter $\theta(t)$ independent of the system dynamics\cite{Proj}. It follows that Theorem \ref{thm:1} is proved if the global boundedness of $e(t)$ is demonstrated. This is achieved in four phases by studying the trajectory of $z(t)$ for all $t\geq t_a$. This methodology was originally proposed in \cite{Megumi} for adaptive control in the presence of time delay.

We begin with suitably chosen finite constants $\bar e$ and $\delta$ such that $\bar e-\delta >0$. The trajectory then has only two possibilities either (i) $|e(t)| < \bar e-\delta$ for all $t\geq t_a$, or (ii) there exists a time $t_a$ at which $|e(t_a)|=\bar e-\delta$. The global boundedness of $e(t)$ is immediate in case (i). We therefore assume there exists a $t_a$ where case (ii) holds.

\begin{enumerate}[label=({\Roman*})]
\item \emph{Entering the Boundary Region}: We start with $\abs{e(t_a)} = \bar{e}-\delta$. We then show that the trajectory enters the boundary region $\B$ at $t_b \in (t_a, t_a+\Delta T_{B})$, and $\BL$ at $t_c>t_b$ where $\Delta T_{B}$ and $t_c$ are finite.
\item \emph{In the Boundary Region, $\BL$}: When the trajectory enters $\B$, the parameter is in the boundary of the projection algorithm; $e$ is shown to be bounded in $\BL$ by making use of the stability property of the underlying linear time-varying system. For $t > t_c$, the trajectory has only two possibilities: either {(i)} $z$ stays in $\BL$ for all $t \geq t_c$, or {(ii)} $z$ reenters $\BU$ at some $t_d>t_c$ where $|e(t_d)| \leq \bar{x}_m$. 
\item \emph{In the Boundary Region, $\BU$}: For $t > t_d$, the trajectory has three possibilities: either {(i)} $z$ reenters $A$ at $t=t_e$, (ii) $z$ stays in $\BU$ for all $t \geq t_d$, or {(iii)} $z$ reenters $\BL$ at $t_f \in (t_d, t_d+\Delta T_{\BL})$ where $\Delta T_{\BL}$ is finite.
\item \emph{Return to Phase I or Phase II}: If case {(i)} from Phase III holds, then the trajectory has only two possibilities: either $|e(t)| < \bar e-\delta, \; \forall t > t_e$ which proves Theorem \ref{thm:1}, or there exists a $t_g>t_e$ such that $|e(t_g)| = \bar e-\delta$ in which case the conditions of Proposition \ref{prop:1} are satisfied with $t_a$ replaced by $t_g$, and Phases I through III are repeated for $t \geq t_g$. If case (ii) from Phase III holds, then the boundedness of $e$ is established for all $t \geq t_d$. If case (iii) holds, then Phases II and III are repeated for $t \geq t_f$. In all cases, $e$ remains bounded throughout.
\end{enumerate}

\subsection{Phase I: Entering the Boundary Region} \label{subsec:Phase1}
We start with $\abs{e(t_a)}=\bar{e}-\delta$. From (\ref{eq:parerr}), it is easy to see that
\begin{equation} \label{eq:boundeta}
\abs{\eta} \leq (k_\eta+1)\theta_{max}(|e|+\bar{x}_m)+(k_\eta+1) r_{max}
\end{equation}
where $k_\eta=\norm{G_\eta (s)}$ and
\begin{equation} \label{eq:xmebar}
\bar{x}_m=\max\limits_{t \geq t_a} \abs{x_m (t)}.
\end{equation}
We define $\bar e$ as
\begin{equation} \label{eq:ebars}
\bar{e} =\max\lbrace {e}_0, e_1\rbrace
\end{equation}
where
\begin{align}
{e}_0&= |x_p(t_a)|+\bar{x}_m+2\delta \label{eq:ebar0} \\
{e}_1&= \frac{1}{2} \left( \bar{c}b_0+\sqrt{\bar{c}^2 b_0^2 +4\bar{c}b_1}\right) \label{eq:ebar1}
\end{align}
with $b_0$ and $b_1$ defined in (\ref{eq:b0}) and (\ref{eq:b1}), $\delta \in (0,\bar{x}_m)$, $\alpha \in (0, \varepsilon_0]$, $c$ in (\ref{eq:varth}), and
\begin{align}
\bar{c}&=\dfrac{2\theta_{max}^{\prime}+\alpha+\dfrac{\varepsilon_0}{c}}{\delta \gamma}.
\end{align}

Phase I is completed by proving the following Proposition:
\begin{prop} \label{prop:1}
Let $z(t_a) \in A$ with $\abs{e(t_a)} = \bar{e}-\delta$ where $\bar{e}$ is given in (\ref{eq:ebars}) and $\delta \in (0,\bar{x}_m)$. Then
\begin{enumerate}[label=({\roman*})]
\item $\abs{e(t)} \leq \bar{e}, \quad \forall t \in [t_a, t_a + \Delta T]$ \label{prop1:i}
\item $z(t_c) \in \BL$ for some $t_c \in (t_a,t_a+\Delta T)$ \label{prop1:ii}
\end{enumerate}
where
\begin{equation} \label{eq:Delta_T}
\Delta T= \frac{\delta}{b_0 \bar{e}+b_1}
\end{equation}
\begin{align}
b_0&=\abs{a_m}+(k_\eta+2)\theta_{max}+\abs{\theta^\star} \label{eq:b0} \\
b_1&= ((k_\eta+2)\theta_{max}+\abs{\theta^\star}) \bar{x}_m +(k_\eta+2)r_{max} \label{eq:b1}
\end{align}
\begin{proof}[Proof of Proposition \ref{prop:1}\ref{prop1:i}]
From (\ref{eq:edot}) and (\ref{eq:boundeta}), it follows that
\begin{equation} \label{eq:b0edot}
\abs{\dot{e}(t)} \leq {b_0 \bar{e}^\prime +b_1}, \quad \forall t \in [t_a, t_a+\Delta T]
\end{equation}
where
\begin{equation} \label{eq:ebarprime}
\bar{e}^\prime=\max\limits_{t \in [t_a, t_a+\Delta T]} \abs{e(t)}.
\end{equation}
We will show below that $\bar e^\prime=\bar e$ which proves Proposition \ref{prop:1}(i). We have that for all $\Delta t \in [0,\Delta T]$,
\begin{align}
\abs{e(t_a+\Delta t)} &\leq \abs{e(t_a)} +\max\limits_{t \in [t_a, t_a+\Delta T]} \abs{\dot{e}(t)}\Delta T \\
&\leq ({\bar{e}-\delta})+({b_0 \bar{e}^\prime +b_1})\frac{\delta}{b_0 \bar{e}+b_1} \label{eq:ebarprime2}
\end{align}
from (\ref{eq:b0edot}), (\ref{eq:ebarprime}), the definition of $\Delta T$, and the choice of $\abs{e(t_a)}$. From (\ref{eq:ebarprime}), with some algebraic manipulations, (\ref{eq:ebarprime2}) can be rewritten as
\begin{equation}
\bar e^\prime \leq \bar e +
\frac{ \delta b_0}{b_0\bar e + b_1} \left( \bar e^\prime - \bar e\right)
\end{equation}
which can be simplified as 
\begin{equation} \label{ineq1}
\left(\bar e^\prime - \bar e\right)
\left(1-b_0\Delta T \right) \leq 0.
\end{equation}
Since $\delta < \bar x_m$, from the definition of $b_0$ and $b_1$, it can be shown that 
\begin{equation}
(1-b_0\Delta T) >0.
\end{equation}
Therefore from (\ref{ineq1}), it follows that 
\begin{equation} \label{ineq2}
\bar e^\prime - \bar e \leq 0.
\end{equation}
From the definition of $\bar e^\prime$ in (\ref{eq:ebarprime}), it follows that only the equality in (\ref{ineq2}) can hold. Hence,
\begin{equation} \label{eq:prop1ebar}
\abs{e(t_a+\Delta t)} \leq \bar{e}, \quad \forall \Delta t \in [0,\Delta T]
\end{equation}
which implies that
\begin{equation} \label{eq:prop1ebar2}
\abs{e(t)} \leq \bar{e}, \quad \forall t \in [t_a,t_a+\Delta T]
\end{equation}
and the proof of Proposition \ref{prop:1}\ref{prop1:i} is complete.
\end{proof}
\begin{proof}[Proof of Proposition \ref{prop:1}(\ref{prop1:ii}]
We note from (\ref{eq:b0edot}) that
\begin{equation} \label{eq:prop1ii}
\abs{e(t)} \geq \abs{e(t_a)}-({b_0 \bar{e} +b_1})\Delta T, \quad \forall t \in [t_a, t_a+\Delta T].
\end{equation}
Since $|e(t_a)| = \bar e-\delta$, (\ref{eq:prop1ii}) can be simplified as
\begin{equation}
\abs{e(t)} \geq \bar{e}-2\delta, \quad \forall t \in [t_a, t_a+\Delta T].
\end{equation}
Since $\bar{e} \geq |x_p(t_a)| + \bar x_m+ 2\delta$ and $\delta < \bar x_m$, it follows that
\begin{equation} \label{eq:xm_ebar}
\bar{e}-2\delta > \bar x_m.
\end{equation}
This in turn implies that $\dot{\theta}(t)$ is negative for all $t\in [t_a,t_a+\Delta T]$ with
\begin{align} \label{eq:prop1thetaderivative}
{\theta(t_a)-\theta(t_a + \Delta t)} & \geq \gamma (\bar{e}-2\delta) (\bar{e}-2\delta- \bar{x}_m){\Delta t}
\end{align}
for all $\Delta t \in [0, \Delta T]$. Defining,
\begin{equation} \label{eq:prop1deltaTinmax}
\Delta T_B =\frac{2\theta_{max} - \varepsilon_0 + \alpha}{\gamma(\bar{e}-2\delta) ( \bar{e}-2\delta- \bar{x}_m)}
\end{equation}
it follows that $z(t_b)$ enters $\B$ at $t_b \in (t_a, t_a+\Delta T_B)$ if $\Delta T_B \leq \Delta T$.

We now show that $z(t_c)$ enters $\BL$ at $t_c< t_a + \Delta T_B^\prime$ for some $\Delta T_B^\prime > \Delta T_B$. It can first be proven that
\begin{equation} \label{eq:prop1proj}
\abs{\Proj(\theta,y)} > c\abs{y}, \quad \forall z \in \BU.
\end{equation}
Then, from (\ref{eq:adaptlawupdate}),
\begin{equation} \label{eq:prop1thetadot2}
-\dot{\theta}(t) > \gamma c (\bar{e}-2\delta) ( \bar{e}-2\delta- \bar{x}_m), \quad \forall t \in T_{\BU}
\end{equation}
where $T_{\BU}$ is defined as
\begin{equation}
T_{\BU}: \lbrace t \; \lvert \; z(t) \in \BU \; \text{and} \; t \in [t_a,t_a+\Delta T] \rbrace.
\end{equation}
Since the distance the trajectory can travel in $\BL$ is bounded by $\xi_0$, the maximum time $z(t)$ spends in $\BU$ can be derived from (\ref{eq:prop1thetadot2}), and we obtain
\begin{equation} \label{eq:deltaTBL}
\Delta T_{\BL}=\frac{(1-c)\varepsilon_0}{\displaystyle {\gamma}{c} (\bar{e}-2\delta) ( \bar{e}-2\delta- \bar{x}_m)}.
\end{equation}
This implies that $z(t_c)$ enters $\BL$ at $t_c \in (t_a,t_a+\Delta T_B^\prime)$ where
\begin{equation}
\Delta T_B^\prime=\Delta T_B + \Delta T_{\BL}
\end{equation}
if $\Delta T_B^\prime \leq \Delta T$, since then (\ref{eq:prop1thetadot2}) is satisfied for all $t \in (t_b,t_c]$. From the choice of $\bar e$ in (\ref{eq:ebars}), we have that
\begin{align} \label{eq:ebar1max}
\bar{e} \geq \frac{1}{2} \left( \bar{c}b_0+\sqrt{\bar{c}^2 b_0^2 +4\bar{c}b_1}\right).
\end{align}
Using algebraic manipulations, it can be shown that (\ref{eq:ebar1max}) implies that $\Delta T_B^\prime \leq \Delta T$. This proves Proposition \ref{prop:1}\ref{prop1:ii}.
\end{proof}
\end{prop}

\subsection{Phase II: In the Boundary Region, $\text{\b{$B$}}_L$} \label{subsec:Phase2}
When the trajectory enters $\BL$, the parameter is in the boundary of the projection algorithm with thickness $\xi_0$; $e(t)$ is shown to be bounded by making use of the underlying linear time-varying system in (\ref{eq:LTVsys}) and (\ref{eq:LTV_A}).

Let $z(t) \in \BL$ for $t\in[t_c, t_d)$. That is, $\theta(t)=-\theta_{max}+\xi(t)$ for $t\in [t_c,t_d)$ with $\xi(t)$ satisfying (\ref{eq:varth_t}) and (\ref{eq:varth}). Since $(c_\eta, A_\eta, b_\eta) \in S_\eta(\bar{a},\theta_{max},\xi_0)$, from (\ref{eq:lyap_negdef}), it follows that
\begin{equation} \label{eq:normxbar}
\norm{\bar{x}(t)} \leq x_0, \quad \forall t \in T_{\BL}
\end{equation}
where $T_{\BL}$ is defined as 
\begin{equation}
T_{\BL}: \lbrace t \; \lvert \; z(t) \in \BL \rbrace.
\end{equation}
Since $\abs{e(t)}\leq \abs{x_p(t)}+\bar{x}_m$ for all $t \in T_{\BL}$ and $\bar{x}=[x_p \; x_\eta^T]^T$, this implies
\begin{equation} \label{eq:ebar2_phase2}
\abs{e(t)}\leq \bar{e}_2, \quad \forall t\in (t_c,t_d)
\end{equation}
where
\begin{equation} \label{eq:ebar2}
\bar{e}_2= x_0 + \bar{x}_m
\end{equation}
which proves boundedness of $e$ in $\BL$. 

We have so far shown that if the trajectory begins in $A$ at $t=t_a$, it will enter the region $\BL$ at $t=t_c$, where $t_c < t_a+\Delta T$, and $\Delta T$ is finite. For $t > t_c$, there are only two possibilities either (i) $z$ stays in $\BL$ for all $t > t_c$, or (ii) $z$ reenters $\BU$ at $t=t_d$ for some $t_d > t_c$. If (i) holds, it implies that (\ref{eq:ebar2_phase2}) holds with $t_d = \infty$, proving Theorem \ref{thm:1}. The following Proposition addresses case (ii):

\begin{prop} \label{prop:2}
Let $z(t) \in \BL$ for $t \in [t_c, t_d)$ and $z(t_d) \in \BU$ for some $t_d > t_c$. Then
\begin{equation}
\abs{e(t_d)} \leq \bar{x}_m
\end{equation}
\begin{proof}
Since $z(t) \in \BL$ for $t \in [t_c, t_d)$ and $z(t_d) \in \BU$ for some $t_d > t_c$, from (\ref{eq:boundaries_lower}), it follows that for any $\Delta t_d \in (0, t_d-t_c]$,
\begin{equation}
\theta(t_d-\Delta t_d) \leq -\theta_{max}+\xi_0,  \quad \theta(t_d) \geq -\theta_{max}+\xi_0. 
\end{equation}
This implies that $\dot{\theta}(t_d) \geq 0$ which in turn implies
\begin{equation} \label{eq:et_d}
\abs{e(t_d)} \leq \bar{x}_m.
\end{equation}
which proves Proposition \ref{prop:2}.
\end{proof}
We note from (\ref{eq:ebar2}) that 
\begin{equation} \label{eq:ebar2_phase2b}
|e(t)| \leq \bar{e}_2, \quad \forall t \in (t_c, t_d]
\end{equation}
which proves boundedness of $e$ in Phase II.
\end{prop}

\subsection{Phase III: In the Boundary Region, $\text{\b{$B$}}_U$} \label{subsec:Phase3}
The boundedness of $e$ has been established thus far for all $t \in [t_a, t_d]$. For $t>t_d$, there  are three cases to consider: either (i) $z$ reenters $A$ at $t=t_e$ for some $t_e > t_d$, (ii) $z$ remains in $\BU$ for all $t \geq t_d$, or (iii) $z$ reenters $\BL$ at $t_f \in (t_d,t_d+\Delta T_{\BL})$ with $\Delta T_{\BL}$ given by (\ref{eq:deltaTBL}). 

We address case (i) in the following Proposition.

\begin{prop} \label{prop:3}
Let $z(t) \in \BU$ for $t \in [t_d, t_e)$ and $z(t_e) \in A$ for some $t_e > t_d$. Then
\begin{equation}
\abs{e(t)} < \bar{x}_m, \quad \forall t \in (t_d, t_e]
\end{equation}
\begin{proof}
Since $z(t) \in \BU$ for $t \in [t_d, t_e)$ and $z(t_e) \in A$ for some $t_e > t_d$, from (\ref{eq:boundaries}), it follows that for any $\Delta t_e \in (0, t_e-t_d]$,
\begin{equation}
\theta(t_e-\Delta t_e) < -\theta_{max}^\prime,  \quad \theta(t_e) \geq -\theta_{max}^\prime.
\end{equation}
This implies that $\dot{\theta}(t)$ is positive, and we obtain
\begin{equation} \label{eq:et_dprime}
\abs{e(t)} < \bar{x}_m, \quad \forall t \in (t_d,t_e]
\end{equation}
which proves Proposition \ref{prop:3}. 
\end{proof}
\end{prop}

We now address case (ii) and (iii). 

We consider suitably chosen finite constants $\bar e_3$ and $\delta$ such that $\bar{e}_3-\delta >0$, and
\begin{equation} \label{eq:ebar3}
\bar{e}_3 =\max \lbrace e_2, e_3 \rbrace
\end{equation}
where
\begin{align}
e_2 &= 2\bar{x}_m+2\delta \label{eq:e2} \\
e_3 &= \frac{1}{2} \left( \bar{c}_2 b_0+\sqrt{\bar{c}_2 ^2 b_0^2 +4\bar{c}_2 b_1}\right) \label{eq:e3}
\end{align}
and
\begin{align}
\bar{c}_2&=\dfrac{(1-c) \varepsilon_0}{\delta \gamma c}.
\end{align}

From (\ref{eq:et_d}) and the definition of $\bar{e}_3$, it follows that 
\begin{equation} \label{eq:phase3_caseii}
|e(t_d)| < \bar{e}_3-\delta.
\end{equation}
If $e(t)$ grows without bound, it implies that there exists $t_d^{\prime} >t_{c}$ such that
\begin{equation} \label{eq:phase3_case3}
|e(t_d^{\prime})|=\bar{e}_3-\delta.
\end{equation}
Hence,
\begin{equation} \label{eq:phase3_case2}
|e(t)|<\bar{e}_3-\delta, \quad \forall t \in [t_d, t_d^{\prime}).
\end{equation}

We show below that if such a $t_d^{\prime}$ exists, then $z(t)$ must enter $\BL$ at $t=t_f$, for some finite $t_f>t_d^{\prime}$.

\begin{prop} \label{prop:4}
Let $z(t) \in \BU$ for all $t \in [t_d,t_f)$, and $\exists t_d^{\prime} \in (t_d,t_f)$ such that $\abs{e(t_d^{\prime})} = \bar{e}_3-\delta$ where $\bar{e}_3$ is given in (\ref{eq:ebar3}) and $\delta \in (0,\bar{x}_m)$. Then
\begin{enumerate}[label=({\roman*})]
\item $\abs{e(t)} \leq \bar{e}_3, \quad \forall t \in [t_d^{\prime}, t_d^{\prime} + \Delta T^\prime]$ \label{prop4:i}
\item $z(t_f) \in \BL$ for some $t_f \in (t_d^{\prime},t_d^{\prime}+\Delta T^\prime)$ \label{prop4:ii}
\end{enumerate}
where
\begin{equation} \label{eq:Delta_Tp}
\Delta T^\prime= \frac{\delta}{b_0 \bar{e}_3+b_1}
\end{equation}
\begin{proof}We note that Proposition (\ref{prop:4}) is identical to Proposition \ref{prop:1} with $t_a$ replaced by $t_d^{\prime}$, $\bar{e}$ replaced with $\bar{e}_3$, and $z(t_d^{\prime}) \in \BU$ which implies $\Delta T_B=0$. Using an identical procedure, we can prove both Proposition \ref{prop:4}\ref{prop4:i} and Proposition \ref{prop:4}\ref{prop4:ii}.
\end{proof}
\end{prop}

We note that if case (ii) holds, it implies that (\ref{eq:phase3_case2}) holds for $t_d^{\prime}=\infty$, which implies that $e(t)$ is globally bounded.

In summary, in Phase III, we conclude that if $z$ enters $\BU$ at $t=t_d$,
\begin{enumerate}[label=({\roman*})]
\item $z$ enters A at $t=t_e$ with $|e(t)| < \bar{x}_m$ for all $t \in [t_d, t_e]$,
\item $z$ remains in $\BU$ for $t \geq t_d$ with $|e(t)| < \bar{e}_3 -\delta$ for all $t \geq t_d$, or
\item $z$ enters $\BL$ at $t=t_f$ for $t_f>t_d$ with $|e(t)| \leq \bar{e}_3$ for all $t \in [t_d, t_f]$.
\end{enumerate}
Therefore, either Phases I and II, or Phases I, II, and III, can be repeated endlessly but with $|e(t)|$ remaining bounded throughout. This is stated in the next section.

\subsection{Phase IV: Return to Phase I or Phase II} \label{subsec:Phase4}
If Proposition \ref{prop:3} is satisfied, then the trajectory has exited the boundary region and entered Region $A$. Therefore, $|e(t)| < \bar{e}-\delta$ for all $t \geq t_e$, in which case the boundedness of $e$ is established, proving Theorem \ref{thm:1}, or there exists a $t_g>t_e$ such that $|e(t_g)|=\bar{e}-\delta$. The latter implies that the conditions of Proposition \ref{prop:1} are satisfied with $t_a$ replaced by $t_g$. Therefore, Phases I through III are repeated for $t\geq t_g$. 

If Proposition \ref{prop:4} is satisfied instead, then  $z$ has reentered $\BL$, in which case Phases II and III are repeated for $t > t_f$. 

By combining (\ref{eq:ebars}) from Phase I, (\ref{eq:ebar2_phase2}), (\ref{eq:ebar2}), and (\ref{eq:ebar2_phase2b}) from Phase II, and (\ref{eq:et_dprime}) and (\ref{eq:ebar3}) from Phase III, we obtain
\begin{equation}
|e(t)| \leq \max \lbrace \bar{e},\bar{e}_2, \bar{e}_3 \rbrace, \quad \forall t \geq t_a
\end{equation}
proving Theorem \ref{thm:1}.
\end{proof}
\end{thm}
\section{Numerical Example} \label{sec:5}
In this section we demonstrate using the counterexample in \cite{Rohrs} as to how the main result in this paper can be used to obtain robust adaptive control in the presence of unmodeled dynamics. We consider the nominal first order stable plant\footnote{$s$ in what follows is a differential operator $d/dt$ and not the Laplace variable.}
\begin{equation}
{x_p}(t)= \frac{2}{\left( s+1 \right)}[u(t)]
\end{equation}
in the presence of highly damped second order unmodeled dynamics, described by (\ref{eq:Geta2ndorder}) with 
\begin{equation} \label{eq:zeta_omega}
\zeta = 0.9912, \quad \omega_n = 15.1327
\end{equation}
and a reference model
\begin{equation}\label{eq:rohrs2}
x_m(t)=\frac{3}{\left( s+3 \right)}[r(t)].
\end{equation}
The adaptive controller is chosen as in (\ref{eq:adaptlaw})-(\ref{eq:thetamajor}) with suitably chosen $\theta_{max}$ and $\varepsilon_0= 0.1\theta_{max}$. The control problem differs slightly from that shown in Fig. \ref{fig:ControlProb} and requires gain compensation on the reference input.

That is, the plant and reference model differ from (\ref{eq:plant}) and (\ref{eq:refmod}) such that
\begin{align}
\dot{x}_p(t) & = a_p x_p(t)+k_p v(t) \label{eq:plant_num} \\
\dot{x}_m(t)& = a_m x_m(t)+k_m r(t) \label{eq:refmod_num}
\end{align}
where 
\begin{equation} \label{eq:num_param}
a_p=-1, \quad k_p=2, \quad a_m=-3, \quad k_m=3.
\end{equation}
The control input is then
\begin{equation} \label{eq:controlin_num}
u(t)=\theta(t)x_p(t)+k_r r(t)
\end{equation}
where $k_r= \sfrac{k_m}{k_p}=1.5$ so as to match the closed-loop adaptive system when no unmodeled dynamics are present ($G_{\eta}(s) \equiv 1$). 

We now show that (\ref{eq:Geta2ndorder}) with (\ref{eq:zeta_omega}) corresponds to $S_\eta(a_p,\theta_{max},\varepsilon_0)$ for suitably chosen $\theta_{max}$ and $\varepsilon_0$. When $\theta(t)=-\theta_{max}$ in (\ref{eq:controlin_num}), the closed-loop adaptive system given by (\ref{eq:plant_num}), (\ref{eq:refmod_num}), (\ref{eq:num_param}), and (\ref{eq:controlin_num}) has a transfer function from $r$ to $x_p$ of the form
\begin{equation} \label{eq:Gcl_2ndorder_numeric}
G_c (s)= \frac{458}{s^3+31s^2+259s+229+458\theta_{max}}.
\end{equation}
In addition to condition \ref{cond1}, the following conditions are neccessary and sufficient for the poles of $G_c(s)$ in (\ref{eq:Gcl_2ndorder_numeric}) to lie in $\mathbb{C}^-$, which are slightly modified versions of \ref{cond2} and \ref{cond3} due to the presence of $k_p$ and $k_m$:
\begin{enumerate}[align=left, label=(A-{\roman*})b]
\setcounter{enumi}{1}
\item $\; \theta_{max} > \displaystyle \frac{a_p}{k_p}$ \label{cond2_num_b}
\item $\theta_{max} < \displaystyle \frac{1}{k_p} \bigg(-4a_p\zeta^2 + \frac{2\zeta a_p^2}{\omega_n} + 2\zeta\omega_n \bigg)$ \label{cond3_num_b}
\end{enumerate}
Therefore, if $\theta_{max}$ is such that 
\begin{equation} \label{eq:ineq1_num}
\displaystyle \frac{a_p}{k_p} < \theta_{max} < \frac{1}{k_p} (a_p+ \bar\theta^\star)
\end{equation}
with $\bar{\theta}^*=35.06$ then conditions \ref{cond2_num_b} and \ref{cond3_num_b} hold. Since $a_p=-1$ and $k_p=2$, if we choose $\theta_{max} = 16.7$, (\ref{eq:ineq1_num}) is satisfied. With $\zeta$ and $\omega_n$ in (\ref{eq:zeta_omega}), \ref{cond1} is satisfied as well.

We now demonstrate the choice of $\varepsilon_0$. Since the closed-loop system in (\ref{eq:Gcl_2ndorder_numeric}) is stable for $\xi_0$ satisfying \ref{cond:3}, a Lyapunov function is chosen as in (\ref{eq:lyap}). It follows from (\ref{eq:lyap2PQ}) and (\ref{eq:lamdamax}) that $Q$ and $P$ are such that
$\lambda_{Q_{min}} =1$ and $\lambda_{P_{max}} = 47773.6$.
Since $\norm{b_\eta} = 229$ from (\ref{eq:Geta2ndorder}) and (\ref{eq:zeta_omega}), we choose $\xi_0$ using \ref{cond:3} such that
\begin{equation} \label{eq:var_num}
\xi_0 = 4.57\cdot 10^{-8}.
\end{equation} 
Condition \ref{cond:2} implies that any $\varepsilon_0$ such that $\xi_0 < \varepsilon_0 < \theta_{max}$ suffices, with the actual value determined between the trade-off between adaptation and numerical accuracy. In the numerical simulations we report below, we chose $\varepsilon_0=0.1 \theta_{max}$.

In summary, $\theta_{max}=16.7$, $\xi_0$ as in (\ref{eq:var_num}), and $\varepsilon_0=1.7$, ensures that the triple $(c_\eta,A_\eta,b_\eta)$ belongs to $S_\eta(a_p,\theta_{max}, \xi_0)$. With these choices, the adaptive controller in (\ref{eq:adaptlawupdate}) and (\ref{eq:controlin_num}) guarantees globally bounded solutions for any initial conditions $x_p(0)$ and $\theta(0)$ with $\norm{\theta(0)} \leq \theta_{max}$ for the Rohrs unmodeled dynamics in (\ref{eq:Geta2ndorder}) and (\ref{eq:zeta_omega}).

\subsection{Simulation Studies}
In this section, we carry out numerical studies of the adaptive system defined by the plant in (\ref{eq:plant_num}) in the presence of unmodeled dynamics in (\ref{eq:Geta2ndorder}) and (\ref{eq:zeta_omega}) with the reference model in (\ref{eq:refmod_num}), the controller in (\ref{eq:controlin_num}), and the adaptive law in (\ref{eq:adaptlawupdate}) with $\theta_{max}=16.7$ and $\varepsilon_0=1.7$. The resulting plant output, $x_p$, reference model output, $x_m$, error, $e$, and $\theta$ are illustrated in Fig. \ref{fig:case_2} for the reference input
\begin{equation} \label{eq:rohrs_sine2}
r(t)= 0.3+1.85\sin(16.1\text{t})
\end{equation}
and initial conditions $x_p(0)=0$ and $\theta(0)=-0.65$. It was observed that all of these quantities became unstable when the projection bound in (\ref{eq:thetamajor}) was removed. It is interesting to note that in this case, only Phases I and II discussed in Section \ref{sec:4} occurred, with Phase I lasting from $t=0$ to $t=1377.5s$ and Phase II for all $t\geq 1377.5s$.
This clearly validates the main result of this paper reported in Theorem 1.
\psfrag{Outputs}[b][][\plab]{Outputs}
\psfrag{Error}[b][][\plab]{$e$}
\psfrag{theta}[b][][\plab]{$\theta$}
\psfrag{time}[t][][\plab]{time $(s)$}
\psfrag{th*}[l][r][\plab]{$\theta^\star$}
\psfrag{Xm}[][r][\pvar]{$x_m(t)$}
\psfrag{Xp}[b][br][\pvar]{$\, x_p(t)$}
\psfrag{H}[c][][\marker]{$\odot$}
\psfrag{X}[][][\marker]{$\bullet$}
\psfrag{-}[r][][\marker][-90]{$\boldsymbol{\scriptscriptstyle{---------}}$}
\psfrag{I}[b][l][\pvar]{\setlength{\tabcolsep}{0pt}\renewcommand{\arraystretch}{10.5}\begin{tabular}{c}$\text{I}_f$\end{tabular}}
\psfrag{Int}[t][][\pvar]{}
\psfrag{Intf}[l][r][\pvar]{$\theta(t_f)$}
\begin{figure}[t]
	\includegraphics[trim=0 0.6cm 0 0.2cm, width=\figwid]{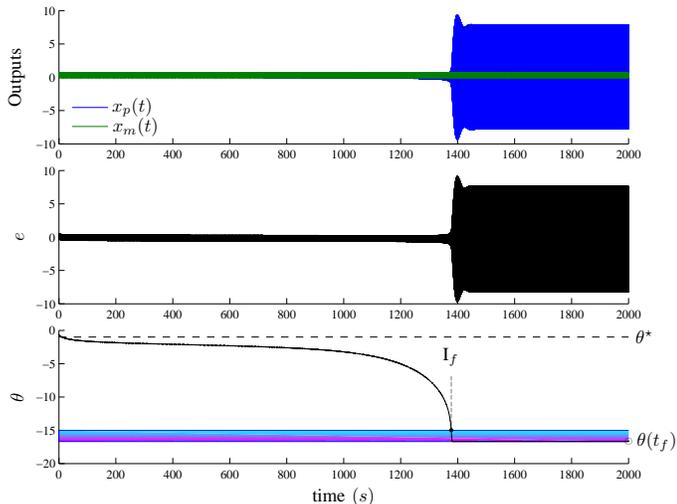} 
	\caption{Simulation of stable plant in the presence of highly damped unmodeled dynamics with adaptive law in (\ref{eq:adaptlawupdate}).}
	\label{fig:case_2}
\end{figure}
In what follows, we carry out a more detailed study of this adaptive system, by only changing the reference input. As the numerical simulations will show, the behavior of the adaptive system, in terms of which of the four phases reported in Section \ref{sec:4} occur, is directly dependent on the nature of the reference input. Four different choices of the reference input are made, and the corresponding behavior are described.
\begin{enumerate}[label=({\roman*})]
\item \underline{$r(t)=0.3+2.0\sin (8t)$:} The error, $e$, and parameter, $\theta$, corresponding to this reference input are shown in Fig. \ref{fig:case_1}. We observe immediately that $|e(t)|< 1$ for all $t\geq 0$. As a result, the trajectory never enters $\B$, eliminating the need for Phases II, III, or IV. Hence, no projection is required in this case. \label{refin:case_1}
\begin{figure}[h]
	\includegraphics[trim=0 0.6cm 0 0.2cm, clip=true,width=\figwid]{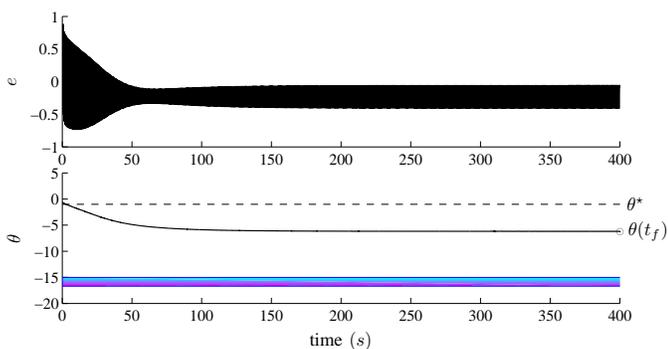} 
	\caption{Simulation results with reference input in \ref{refin:case_1}.}
	\label{fig:case_1}
\end{figure}
\item \underline{$r(t)$: A pulse for the first one second}. That is,
\begin{equation}
r(t) = \begin{cases} 12 &0 \leq t \leq 1s \\ 
0 & t > 1s\end{cases} 
\end{equation}
The corresponding trajectories are shown in Fig. \ref{fig:case_3}, which illustrate that Phase I occurs for $0\leq t\leq 0.9s$, and Phase II for $0.9s \leq t < 1.0s$. The trajectory exits the boundary region at $t_e= 1.0s$, demonstrating Phase III. Phase I is repeated, and the trajectory reenters $\B$ at $t_b=1.3s$, demonstrating Phase IV. The trajectory then settles in $\B$ for all $t\geq 1.3s$. \label{refin:case_3}
\psfrag{Error}[][][\plab]{$e$}
\psfrag{time}[][][0.6]{time $(s)$}
\psfrag{thi}[t][][\pvar]{}
\psfrag{thf}[l][r][\pvar]{$\theta(t_f)$}
\psfrag{-}[r][][\marker][90]{$\boldsymbol{\scriptscriptstyle{----}}$}
\psfrag{+}[l][][\marker][90]{$\boldsymbol{\scriptscriptstyle{-----------}}$}
\psfrag{/}[r][][\marker][90]{$\boldsymbol{\scriptscriptstyle{----}}$}
\psfrag{I}[tl][][\pvar]{\setlength{\tabcolsep}{-1.75pt}\renewcommand{\arraystretch}{3.1}\begin{tabular}{c}$\text{I}_f$\end{tabular}}
\psfrag{II}[bl][][\pvar]{\setlength{\tabcolsep}{-3.5pt}\renewcommand{\arraystretch}{12.0}\begin{tabular}{c}$\text{II}_f$\end{tabular}}
\psfrag{III}[tl][][\pvar]{\setlength{\tabcolsep}{1.75pt}\renewcommand{\arraystretch}{2.4}\begin{tabular}{c}$\text{III}_f$\end{tabular}}
\psfrag{Int}[t][b][\intc]{$z(0)$}
\psfrag{Intf}[tr][b][\intc]{$z(t_f)$}
\psfrag{Y}[][][\marker]{$\bullet$}
\begin{figure}[h]
\psfrag{theta}[t][][\plab]{$\theta$}
\psfrag{Error}[t][][\plab]{$e$}
	\includegraphics[trim=0 0.6cm 0 0.2cm, clip=true,width=\figwid]{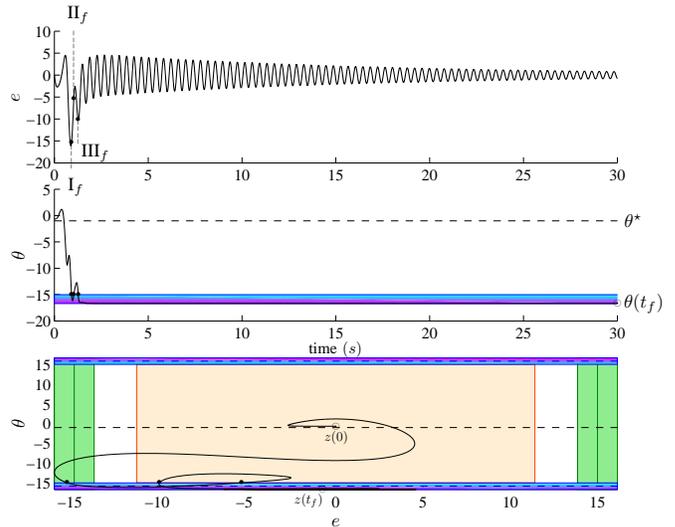} 
	\caption{Simulation results with reference input in \ref{refin:case_3}.}
	\label{fig:case_3}
	\end{figure}
\item  \label{refin:case_4} \underline{$r(t)=10 \; \forall t$:} Fig. \ref{fig:case_4} illustrates the corresponding limit cycle behavior of the trajectory. We observe that the trajectory first enters $\B$ at $t_{b}=1.80s$. Phase II then occurs for $1.80s \leq t < 9.82s$. Phase III occurs for $t_e=9.82s$, and then Phase I is repeated with the trajectory reentering $\B$ at $t_b=9.84s$. Phases I through III are repeated for all $ t \geq 9.84$, demonstrating Phase IV, a limit cycle behavior. The points at which the trajectory enters $\B$ (i.e. Phase II) are shown in orange, and the points at which the trajectory exits $\B$ (i.e. Phase III) are shown in purple.
\begin{figure}[h]
\psfrag{-}[l][][\marker][90]{$\boldsymbol{\scriptscriptstyle{---}}$}
\psfrag{+}[l][][\marker][90]{$\boldsymbol{\scriptscriptstyle{--------}}$}
\psfrag{/}[r][][\marker][90]{$\boldsymbol{\scriptscriptstyle{-----}}$}
\psfrag{X}[][][\marker]{\color[rgb]{0.5781,0,0.8242}$\bullet$}
\psfrag{Y}[][][\marker]{\color[rgb]{1,0.5469,0}$\bullet$}
\psfrag{I}[bl][][\pvar]{\setlength{\tabcolsep}{-1.75pt}\renewcommand{\arraystretch}{3.2}\begin{tabular}{c}$\text{I}_{f_1}$\end{tabular}}
\psfrag{II}[bl][][\pvar]{\setlength{\tabcolsep}{-3.5pt}\renewcommand{\arraystretch}{8.5}\begin{tabular}{c}$\text{II}_{f_1}$\end{tabular}}
\psfrag{III}[tl][][\pvar]{\setlength{\tabcolsep}{1.75pt}\renewcommand{\arraystretch}{2.65}\begin{tabular}{c}$\text{III}_{f_1}$\end{tabular}}
\psfrag{theta}[t][][\plab]{$\theta$}
\psfrag{Error}[t][][\plab]{$e$}
	\includegraphics[trim=0 0.6cm 0 0.2cm, clip=true,width=\figwid]{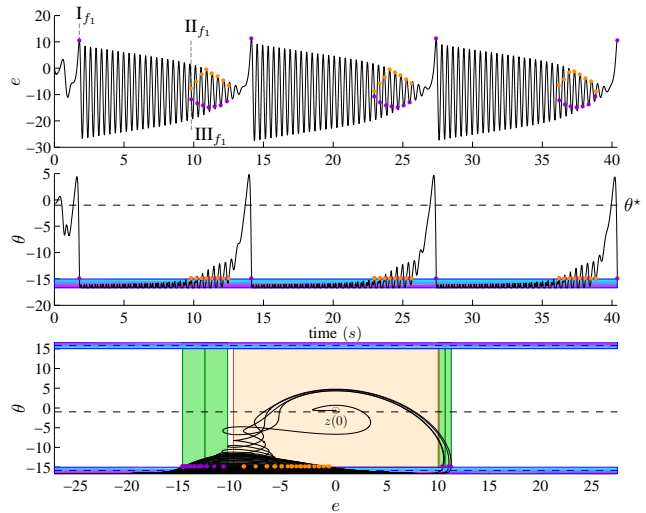} 
	\caption{Simulation results with reference input in \ref{refin:case_4}.}
\label{fig:case_4}
\end{figure}
\item \underline{$r(t)=2 \sin (\omega_0 t)$}, with $\omega_0$ undergoing a continuous sweep from 16.1 rad/s to 2 rad/s over a sixty second interval. The resulting trajectories are shown in phase-plane form in Fig. \ref{fig:case_5}b, for six different initial conditions labeled 1 through 6. It is observed that the trajectory behaves differently for each initial condition. Initial conditions 1 through 3 resulted in trajectories that remained in Phase II for all $t \geq t_b$. Initial condition 4 led to a trajectory with a finite number of occurrences of Phases I through III and finally settled in Region $A$. Initial conditions 5 and 6 stayed in Region $A$ for all $t \geq 0$. All steady state values are labeled as $1f$ through $6f$. \label{refin:case_5}
\begin{figure}[h]
\psfrag{time}[][][0.6]{time $(s)$}
\psfrag{theta}[t][][\plab]{$\theta$}
\psfrag{Error}[t][][\plab]{$e$}
\psfrag{H}[][][\marker]{$\odot$}
\psfrag{S}[][][\marker]{}
\psfrag{X}[][][\marker]{}
\psfrag{Y}[][][\marker]{$\odot$}
\psfrag{th*}[bl][r][\plab]{$\theta^\star$}
\psfrag{1H}[b][t][\intc]{$1$}
\psfrag{2H}[b][t][\intc]{$2$}
\psfrag{3H}[b][t][\intc]{$3$}
\psfrag{4H}[b][t][\intc]{$4$}
\psfrag{5H}[b][t][\intc]{$5$}
\psfrag{6H}[b][t][\intc]{$6$}
\psfrag{4F}[l][r][\intc]{$4f$}
\psfrag{5F}[tl][r][\intc]{$5f$}
\psfrag{6F}[l][r][\intc]{$6f$}
\psfrag{2F}[l][r][\intc]{$1f,2f,3f$}
\psfrag{1S}[l][r][\intc]{}
\psfrag{3S}[l][r][\intc]{}
\psfrag{4S}[][][0.3]{\color[rgb]{1,1,1}$4f$}
\psfrag{5S}[][b][0.3]{\color[rgb]{1,1,1}$5f$}
\psfrag{6S}[][][0.3]{\color[rgb]{1,1,1}$6f$}
\psfrag{2S}[t][][0.3]{\color[rgb]{1,1,1}$1f,2f,3f$}
	\includegraphics[trim=0 0.6cm 0 0.2cm, clip=true,width=\figwid]{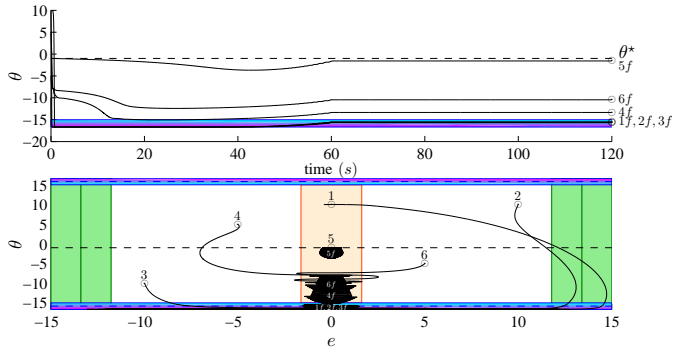}
	\caption{Simulation results with reference input in \ref{refin:case_5}.} 
	\label{fig:case_5}
\end{figure}
\end{enumerate}

\section{Summary}
In this paper, robust adaptive control of scalar plants in the presence of unmodeled dynamics is investigated. It is shown through analytic methods and simulation results that implementation of a projection algorithm in standard adaptive control law achieves global boundedness of the overall adaptive system for a class of unmodeled dynamics. The restrictions on the class of unmodeled dynamics and the projection bounds are explicitly calculated and demonstrated using the Rohrs counterexample.

\bibliographystyle{IEEEtran}
\bibliography{HHussain}

\end{document}